\documentclass[10pt,journal,compsoc]{IEEEtran}

\usepackage{amsthm}
\usepackage{amsmath}
\usepackage{amsfonts}
\usepackage{amssymb}
\usepackage{braket}
\usepackage{verbatim}
\usepackage{physics}
\usepackage{enumerate}
\usepackage{graphicx, subfigure}
\usepackage{array}
\usepackage{fixltx2e}
\usepackage{stfloats}
\usepackage{url}
\usepackage{setspace}
\usepackage{lineno}
 \usepackage[]{algorithm2e}
\usepackage{bm}
\usepackage{booktabs}

\newtheorem{theorem}{Theorem}[section]

\theoremstyle{definition}
\newtheorem{defn}{Definition}[section]
\begin{document}

\title{Randomized Load-balanced Routing for Fat-tree Networks}

\author{Suzhen ~Wang, Jingjing Luo, Bruce Kwong-Bun Tong, and Wing~ Shing ~Wong,~\IEEEmembership{Fellow,~ IEEE}
\IEEEcompsocitemizethanks{\IEEEcompsocthanksitem  S. Wang,  J. Luo,  B. Tong, and W. S. Wong are with the Department
of Information Engineering, The Chinese University of Hong Kong, 
E-mail: ws012@ie.cuhk.edu.hk, wswong@ie.cuhk.edu.hk.
}
}


\IEEEtitleabstractindextext{
\begin{abstract}
Fat-tree networks have been widely adopted to High Performance Computing (HPC) clusters and to Data Center Networks (DCN). These parallel systems usually have a large number of servers and hosts, which generate large volumes of highly-volatile traffic. 
Thus, distributed load-balancing routing design becomes critical to achieve high bandwidth utilization, and low-latency packet delivery. 
Existing distributed designs rely on remote congestion feedbacks to address congestion, which add overheads to collect and react to network-wide congestion information. 
In contrast, we propose a simple but effective load-balancing scheme, called Dynamic Randomized load-Balancing (DRB), to achieve network-wide low levels of path collisions through local-link adjustment which is free of communications and cooperations between switches.
First, we use D-mod-k path selection scheme to allocate default paths to all source-destination (S-D) pairs in a fat-tree network, guaranteeing low levels of path collision over downlinks for any set of active S-D pairs.
Then, we propose Threshold-based Two-Choice (TTC) randomized technique to balance uplink traffic through local uplink adjustment at each switch.
We theoretically show that the proposed TTC for the uplink-load balancing in a fat-tree network have a similar performance as the two-choice technique in the area of randomized load balancing.
Simulation results show that DRB with TTC technique achieves a significant improvement over many randomized routing schemes for fat-tree networks. 
\end{abstract}


\begin{IEEEkeywords}
Fat-tree networks, dynamic load-balancing routing, threshold-based two-choice
\end{IEEEkeywords}
}

\IEEEdisplaynontitleabstractindextext
\IEEEpeerreviewmaketitle
\maketitle
\IEEEraisesectionheading{\section{Introduction}\label{sec:introduction}}
\IEEEPARstart{A}{} fat-tree is a folded version of a Clos network, and has numerous desirable features, such as scalability and high path diversity \cite{ math_ICN, fat-trees}.
Hence, fat-trees are widely adopted for applications ranging from Network-on-Chip \cite{Noc1, Noc2} (NoC), to High-Performance Computing (HPC) clusters \cite{changsha}, to Data Center Networks (DCN) \cite{al_scalable, facebook, portland, jupiter}. 
The high-path diversity of a fat-tree \textemdash there exist many alternative paths between a  Source-Destination (S-D) pair \textemdash is employed to realize load-balanced routing in several manners: centralized, fine-granularity, and congestion-aware.
Centralized load-balancing mechanisms \cite{hedra, openFlow} rely on a centralized control to make path decisions on the basis of a global view of the network.
However, these mechanisms usually face vital issues such as scalability and in-time response to congestion and traffic dynamics.
Fine-granularity mechanisms \cite{flare, presto}  \textemdash wherein a long flow is broken into smaller flows, which are then assigned to different paths \textemdash are less practical because they must modify protocols of the transportation layer. Congestion-aware mechanisms \cite{conga_2, global_conga, DiFS, ECN} \textemdash which aim to seek lightly loaded paths for the traffic based on congestion feedbacks \cite{DiFS, ECN} or on collected congestion information of all possible paths \cite{conga_2, global_conga} \textemdash inevitably add overhead and costs for network-wide information collection and storage.
In this paper, we propose a distributed local-adaptive scheme for fat-tree DCNs to achieve low-levels of path collision across all links through local-uplink adjustment on the basis of D-mod-k deterministic path selection scheme \cite{deterministic, oblivious_D_Mod_k, d_mod4} (to balance downlink traffic) and a proposed threshold-based two-choice randomized technique (to balance uplink traffic). 
The proposed scheme requires a minimum effort since it is distributed and is free of cooperation and communication between switches, and thus can be easily applied to systems such as NoC, HPC and DCN.

Routing in fat-tree can be partitioned into two stages: \textit{upward} (wherein traffic is forwarded upwards) and \textit{downward} (wherein traffic is forwarded downwards).
We refer to links transversed in the upward and downward stage as  \textit{uplinks} and \textit{downlinks} respectively.
The fat-tree topology determines that the upward routing is fully adaptive, whereas the downward routing is deterministic because each switch can select any of its uplinks to forward the traffic upwards and can use only one downlink to  forward the downward traffic to a specific destination. We shall elaborate on these topological routing properties in the next section.
Hence, it is easy to locally balance the uplink traffic but is impossible to locally balance the downlink traffic. 
In contrast, D-mod-k \cite{deterministic, oblivious_D_Mod_k, d_mod4}, which is a deterministic destination-based path selection scheme, yields desirable performance in the downward routing stage. 
To better understand D-mod-k, we mathematically explore the downward routing performance of D-mod-k and prove that it yields the lowest level of path collision over downlinks for any set of active S-D pairs.
However, traffic routed on D-mod-k paths may suffer heavy uplink congestion.

Two-choice technique \cite{TwoChoices, balanced_allocations}  \textemdash a simple but powerful randomized technique for balanced allocation \textemdash can be employed to balance uplink traffic at each switch. 
The power of two-choice technique is described as follows.
Sequentially throw $n$ balls into $n$ bins (assume $n$ is large),  each ball can be thrown into a randomly selected bin (one-choice), or it can be thrown into a less-loaded bin of two randomly selected bins (two-choice). 
At the end of the process, with high probability the maximum load of a bin is $\Theta(\frac{\ln n}{\ln \ln n})$ under the one-choice case and is about $\Theta(\frac{\ln \ln n}{\ln 2})$ under the two-choice case.
This result reveals that two-choice achieves an exponential improvement over one-choice. If d ($\ge 2$) choices are provided for each ball, then the maximum load of a bin is $\Theta(\frac{\ln \ln n}{\ln d})$, just a fractional improvement over two-choice. 
Yu and Deng \cite{openFlow} proposed a scheme that applied the multiple-choice technique to randomized path allocation; the controller assigns each newly-activated S-D pair with the least-congested path among several randomly selected paths. However, Richa at al. \cite{power_two_survey} (see Section 3.4) pointed out that an exponential improvement on the maximum link load cannot be achieved by the path-level multiple-choice scheme.
In contrast, Ghorbani et al. \cite{micro}  proposed a link-level two-choice scheme, called Micro, for dynamic uplink allocation, where each switch always directs newly arrival traffic to the less-congested uplink out of two randomly selected ones. Simulations in \cite{micro} showed that Micro outperforms Conga \cite{conga_2} (a popular congestion-aware load-balancing routing scheme) and Presto \cite{presto} (a fine-granularity load-balancing mechanism). 
However, Micro cannot address downlink congestion because downward routing in fat-tree networks is deterministic. 


Based on the above observations, we propose a Dynamic Randomized load-Balancing (DRB) scheme which first sets paths determined by D-mod-k as default paths and then attempts to balance uplink traffic through local-uplink adjustment. D-mod-k paths guarantee low levels of downlink congestion, whereas local-uplink adjustment contributes to alleviating uplink congestion.
We observe that local-uplink adjustment based directly on the two-choice technique causes too much redirected traffic, which may induce heavy downlink congestion because redirected traffic is no longer routed on D-mod-k paths. To address this, we attempt to limit redirection operation by introducing a threshold to the two-choice technique, which turns out to be effective and achieves a substantial improvement over direction application of two-choice.
We refer to the two-choice technique with a threshold as \textit{ Threshold-based Two-Choice} (TTC). 
It is known through theoretical analysis that the two-choice technique well balances load among different queues in the supermarket model \cite{balanced_allocations, TwoChoices}. 
However, it is unclear and is of interest to investigate that whether or not the introduced threshold will destroy the effectiveness of the two-choice technique.
A theoretical contribution of this paper is that we prove that the power of TTC is comparable to that of the two-choice technique --- they both achieve doubly  exponential queue-size distribution at the equilibrium state in the supermarket model.
Through numerical experiments, we show that  DRB achieves the best performance with a substantially improvement observed especially in heavy-traffic scenarios when compared with many other routing schemes such as  Micro \cite{micro},  Valiant Load Balancing (VLB) \cite{valiant}, and D-mod-k.

The rest of the paper is organized as follows:
In section 2, we  introduce the fat-tree topology and its topological routing properties, and prove the downward routing properties of  D-mod-k.
In section 3, we describe DRB and TTC as well as two typical applications of DRB. 
In section 4, we present the numerical experiment results of two typical applications. 
In section 5, we theoretically analyze the impact of the introduced threshold to the two-choice technique. 
In section 6, we conclude the paper.

\section{ Fat-tree Topology and its Routing properties }
In this paper, we consider a fat-tree topology that has been widely applied to data center network. 
A fat-tree topology has multiple switches layers and one host layer  (see Fig. \ref{tree1}).
We use the notation $\mathcal{F}(\ell, d)$ to indicate a fat-tree with $\ell$ layers of identical 2d-port switches.
We refer to switches at the top-most layer as \textit{core} switches, and refer to switches at the rest layers as \textit{intermediate} switches. 
In a fat-tree $\mathcal{F}(\ell, d)$, there are $2d^\ell$ hosts,  $d^{\ell -1}$ core switches, and $2d^{\ell-1}$ intermediate switches at each layer. 
We label those hosts from $0$ to $2d^\ell-1$, those intermediate switches at each layer from $0$ to $2d^{\ell-1} -1$, and those core switches from $0$ to $d^{\ell-1} -1$.

We partition switch ports into two classes: \textit{up-ports} and \textit{down-ports}.
Up-ports are those ports that are used to connect to switches at a higher layer; down-ports are those ports that are used to connect to switches or hosts at a lower layer. In view of this, core switches have only down-ports.
Each port can send and receive traffic at the same time. 
From left to right, we label down-ports of a core switch from $0$ to $\small 2d-1$, and label up-ports and down-ports of an intermediate switch from $0$ to $\small d-1$ respectively. 
In Fig. \ref{tree2}, we highlight  up-port and down-port labels of each switch in red and in black. 
In what follows, we indicate a port of each switch by its label and denote by $p_{dn}^m$ and $p_{up}^m$ a down-port and an up-port at layer $m$ respectively.

\begin{figure}[htbp]
\centering\includegraphics[width= 3.6 in]{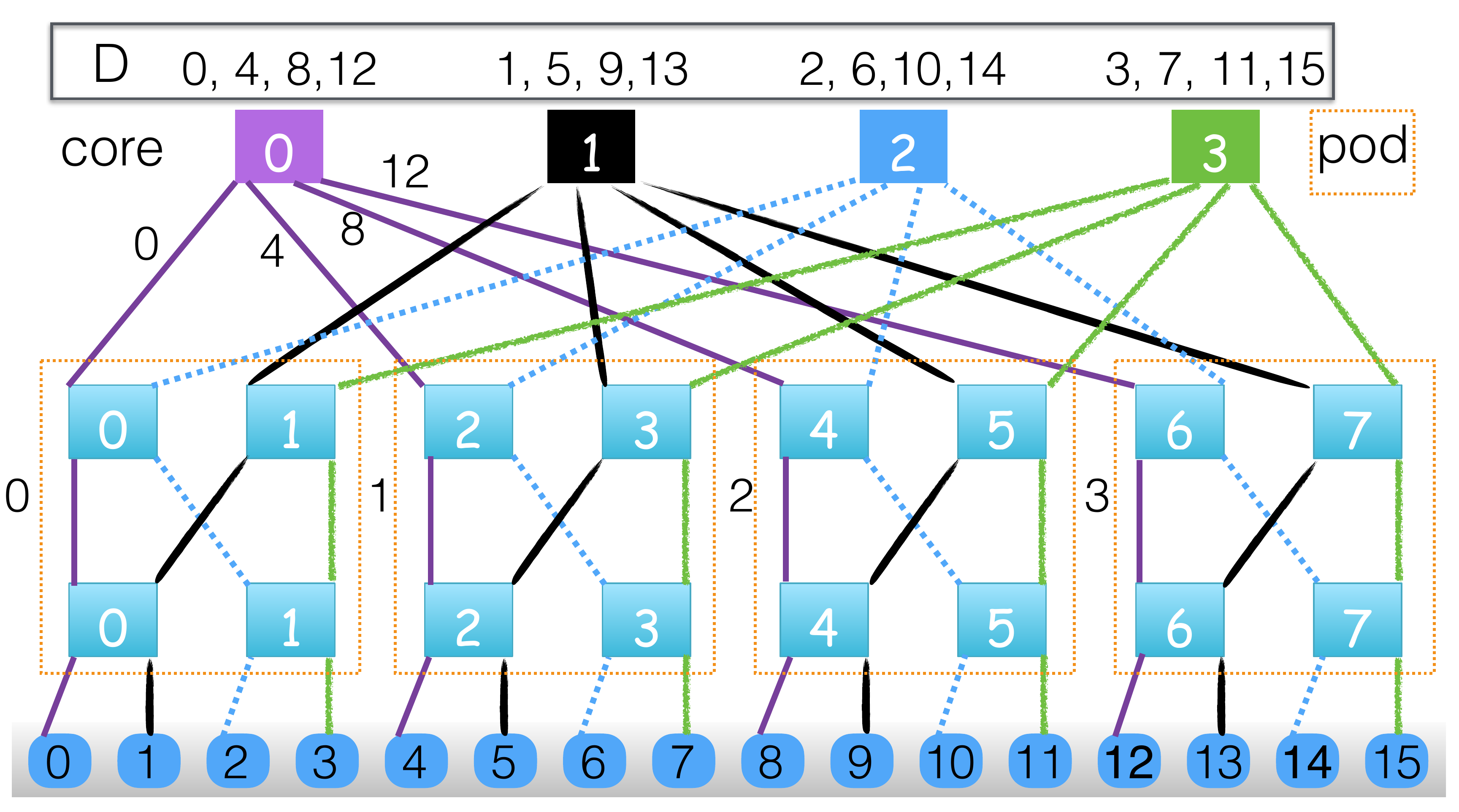} 
\caption{ This figure presents the topology of a fat Tree $\mathcal{F}(3, 2)$, where the bottom-most layer is the host layer, the rests are switch layers. }
\label{tree1} 
\end{figure}

Previous researchers have used integer vectors to encode labels of switches and hosts \cite{karyntrees, oblivious_sigmetrics, DCF} on the basis of the fat-tree topology. These integer vectors in turn characterize the fat-tree topology and define the way that switches and hosts are physically linked.
Given a host whose label is $l_h$, we use an integer  vector $ h=(h_\ell, ..., h_{1}) \in  \mathbb{Z}_{2d} \times \mathbb{Z}^{\ell-1}_d $ to encode this label, which satisfies $l_h = \sum_{i=1}^{\ell} h_i d^{i-1}$.
We index switch layers from bottom to top with numbers ranging from 1 to $\ell$.
We use an integer vector $s^m = (s^m_{\ell-1}, ..., s^m_{1})\in \mathbb{Z}_{2d} \times \mathbb{Z}^{\ell-2}_d$ to encode the label of an intermediate switch at layer $m$ ($m\in [1, \ell-1]$), and use $s^\ell = (s^\ell_{\ell-1}, ..., s^\ell_1)\in \mathbb{Z}^{\ell-1}_d$ to encode the label of a core switch. It has $ l^m_s = \sum_{i=1}^{\ell-1} s^m_i d^{i-1}$ for all $m\in [1, \ell]$, where $l^m_s$ indicates the label of a switch at layer $m$.
We refer to these integer vectors as \textit{label codes}. 
Fig. \ref{tree1} and Fig. \ref{tree2} present labels and label codes for hosts and switches in a $\mathcal{F}(3, 2)$ respectively.
In what follows, we indicate a host or a switch either by its label or its label code.

Direction connection of elements in $\mathcal{F}(\ell, d)$ can be formally defined by label codes. 
Consider switches at the 1st layer. Down-ports of these switches are connected directly to hosts. 
More specifically, a down-port $p_{dn}^1$ of a switch $s^1$ is connected to a host $h$ if $ h = (p_{dn}^1, s^1)$.
Two switches in layers $m-1$ and $m$ $(2 \le m\le \ell)$ respectively are connected directly to each other if their label codes satisfy 
\begin{equation}
\label{switch_connection}
\small
\begin{cases}
s^{m-1}_{-(m-1)} = s^{m}_{-(m-1)}, &\\
p_{up}^{m-1}= s^m_{m-1}, p_{dn}^{m} =  s^{m-1}_{m-1}, &
 \end{cases}
\end{equation}
where $p_{up}^{m-1}$ and $p_{dn}^m$ indicate the up-port label of switch $s^{m-1}$ and the down-port of switch $s^{m}$ respectively. Those two ports are designated to establish the above connection.
Let $x = (x_1, ..., x_n)$, $x_{-i}$ indicates a new vector derived by eliminating the $i$-th component of $x$, that is $x_{-i} = (x_1, ..., x_{i-1}, x_{i+1}, ..., x_n)$ with $i \in [1, n]$. 
Take $\mathcal{F}(2, 3)$ in Fig. \ref{tree2} for example. The two connected switches $s^1 = (0, 0)$ and $ s^2 = ( 0, 1)$ satisfy $s^1_{2} =s^{2}_2 = 0$, where $s^1$ uses up-port with label $1$ and $s^2$ uses down-port with label $0$ to establish the connection.

\begin{figure}[htbp]
\centering\includegraphics[width= 3.6 in]{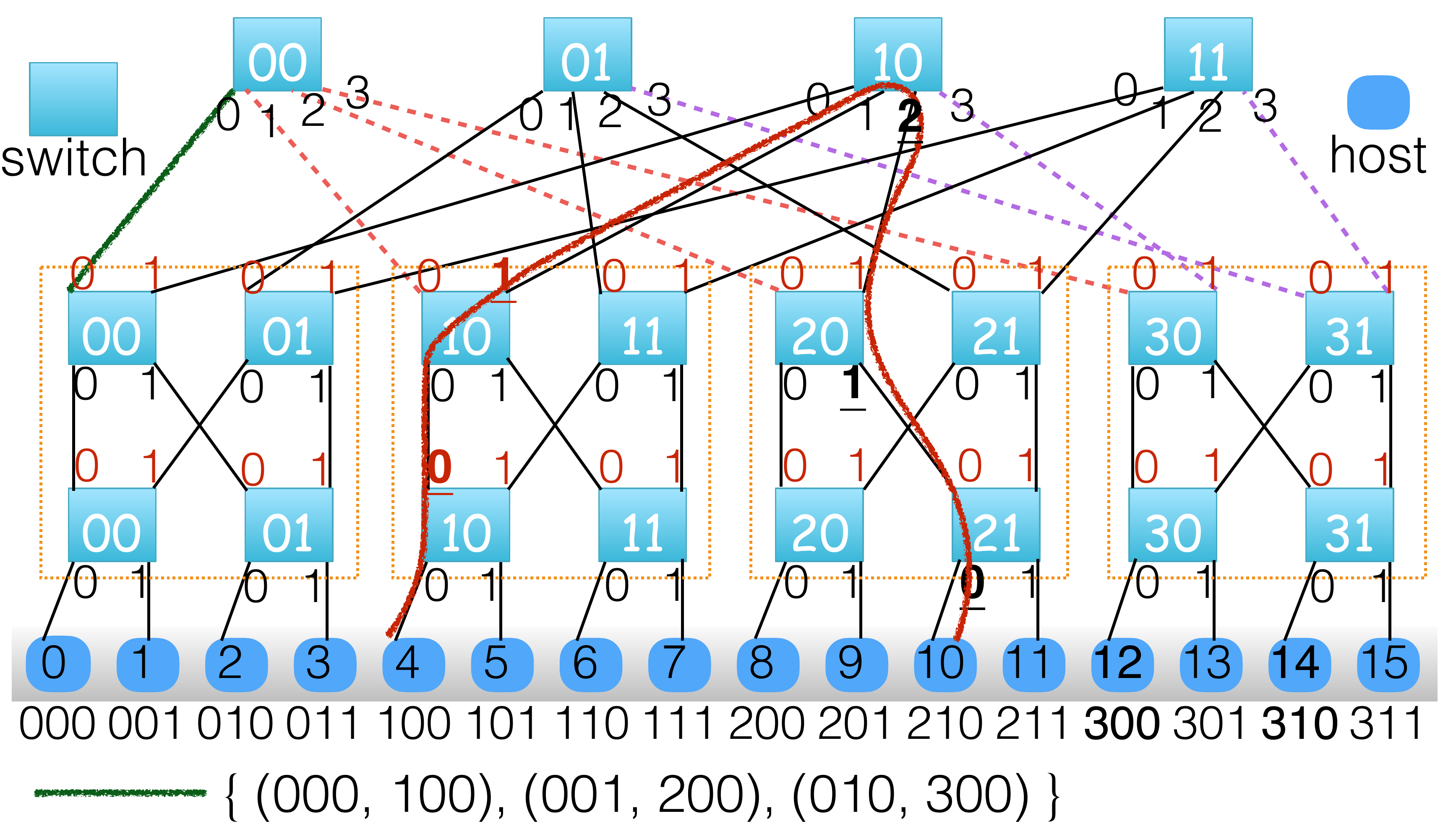} 
\caption{ This figure shows label codes of hosts and switches in a fat-tree $\mathcal{F}(3, 2)$. The labels of up- and down- ports of each switch, which are indicated in red and in black respectively. }
\label{tree2} 
\end{figure}


As mentioned before, routing in a fat-tree proceeds in two stages: upward and downward routing.
We refer to a switch where the routing direction is changed from upwards to downwards as a \textit{transition} switch. 
We also refer to the paths from the source hosts to the transition switches as \textit{upward paths} and the paths from the transition switches to the destination hosts as \textit{downward paths}. 
We denote by $(h^s, h^d)$ an S-D pair with $h^s$ and $h^d $ being its source and destination  respectively.
Define distance between $h^s$ and $h^d$ by
\begin{math}
D(h^s, h^d) = \max_i \{ i\mid h^s_i  \neq h^d_i\}. 
\end{math}  
If the distance of two hosts is $k$, then traffic of this pair needs to transverse at least $k-1$ layers to reach a transition switch at layer $k$ under the shortest path routing principle. 

Given an S-D pair, any of its upward paths can be defined by a sequence of up-ports it transverses; any of its downward paths can be defined by a sequence of down-ports it transverses. 
Let $P_{up}(h^s, h^d, k)$ and $P_{dn}(h^s, h^d, k)$ denote an upward and a downward path of pair $(h^s, h^d)$ with distance $k$, then we have
\begin{equation}
\label{path_definition}
\begin{cases}
P_{up}(h^s, h^d, k) = (p_{up}^1,..., p_{up}^{k-1}),&\\
P_{dn}(h^s, h^d, k) = (p_{dn}^{k}, ..., p_{dn}^1). &
\end{cases}
\end{equation}
For example, the upward path and the downward path of the highlighted path in Fig. \ref{tree2} from host 4 to host 10 are represented as follows.
\begin{displaymath}
\begin{cases}
P_{up} \left((1, 0, 0), (2, 1, 0), 3\right)  = (0, 1), &\\
P_{dn} \left((1, 0, 0), (2, 1, 0), 3\right) = (2, 1, 0).
\end{cases}
\end{displaymath}

An essential routing property of the fat-tree network is that the upward routing is fully adaptive, whereas the downward routing is deterministic. A mathematical explanation of the property is provided as follows. 
Every sequence in $\mathbb{Z}^{k-1}_d$ $(k\in [1, \ell] )$ uniquely determines an upward path for any two pairs with distance $k$. Once a transition switch is reached, the downward path is then uniquely determined by $h^d$ in the form of 
\begin{equation}
\label{determinisitc_dwn}
P_{dn}(h^s, h^d, k)= (h^d_{k}, ..., h^d_1).
\end{equation}
We provide a proof for result in equation (\ref{determinisitc_dwn}) in Appendix A.

\subsection{ D-mod-k Path Selection Scheme}
D-mod-k is a desirable upward path selection scheme, which selects an upward path in the following manner.
\begin{displaymath}
P_{up}(h^s, h^d, k) = (h^d_1,..., h^d_{k-1}), 
\end{displaymath}
which guarantees that traffic to different destinations is routed onto link-disjoint downward paths.
In fact, D-mod-k establishes a mapping from destinations of pairs to transition switches such that  downward paths for pairs with different destinations are link-disjoint. 
For example, different core switches in Fig. \ref{tree1} are mapped by D-mod-k to disjoint sets of destinations (e.g., core switch 0 is mapped to  destinations $\{0, 4, 8, 12\}$). 
Then, as one can check with Fig. \ref{tree1}, D-mod-k downward paths from different core switches to different destinations are link disjoint. 
We summarize the D-mod-k routing property with rigorous proof in the the following theorem. It worths mentioning that we are the first to provide a mathematical verification for the downward routing property of D-mod-k.

\begin{theorem}
\label{thereom_dcf2}
D-mod-k guarantees link-disjoint downward paths for pairs with different destinations in a fat-tree network.
\end{theorem} 
The proof is provided in Appendix B. 
An interpretation of this theorem is as follows.
Given an arbitrary set of S-D pairs, we classify those pairs into different categories based on their destinations where pairs with the same destinations are classified into the same categories and pairs with different categories are classified into different categories.
According to Theorem \ref{thereom_dcf2}, D-mod-k downward paths for any two pairs from different categories are link-disjoint; D-mod-k downward paths for pairs from the same category use common downlinks to reach the same destination. 
Observe that paths of pairs in the same category must collide in some downlink of 1st-layer switches.  
We thus claim that  D-mod-k minimizes the level of downward path collision.

Nevertheless, D-mod-k may perform poor in the upward routing stage. 
Routing in the upward stage under D-mod-k is similar to any deterministic routing in a Butterfly network. 
It also has been shown that there exist certain permutation patterns (wherein each host is a source and a destination of at most one S-D pair) such that about $\sqrt{N}$ pairs use the same links under any deterministic routing scheme for fat-tree networks \cite{oblivious_sigmetrics}. This result implies that D-mod-k routing may induce heavy uplink congestion. 
To address this, we introduce randomized uplink-load balancing scheme in the next section.

\section{Threshold-based Two-choice Scheme}

\begin{figure}
\centering 
\subfigure[The scheduler sends the traffic to the 1st choice since the load difference between the two choices does not exceed the threshold $T (c)= 1$. ] { \label{fig:a} 
\includegraphics[width=0.8\columnwidth]{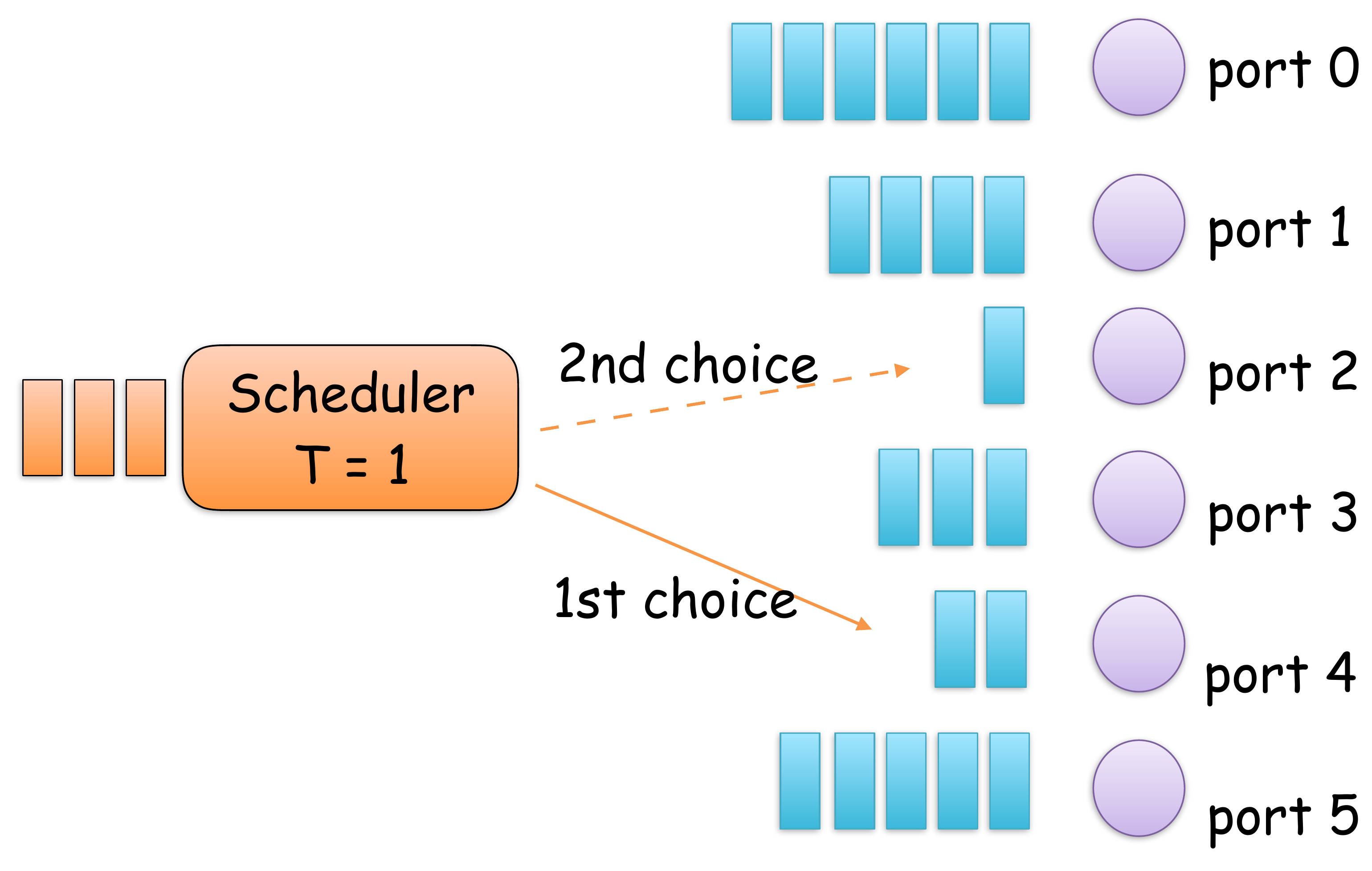} 
} 
\subfigure[The scheduler sends the traffic to the 2nd choice since the 2nd choice is less loaded with the load difference larger than T (c).] { \label{fig:b} 
\includegraphics[width=0.8\columnwidth]{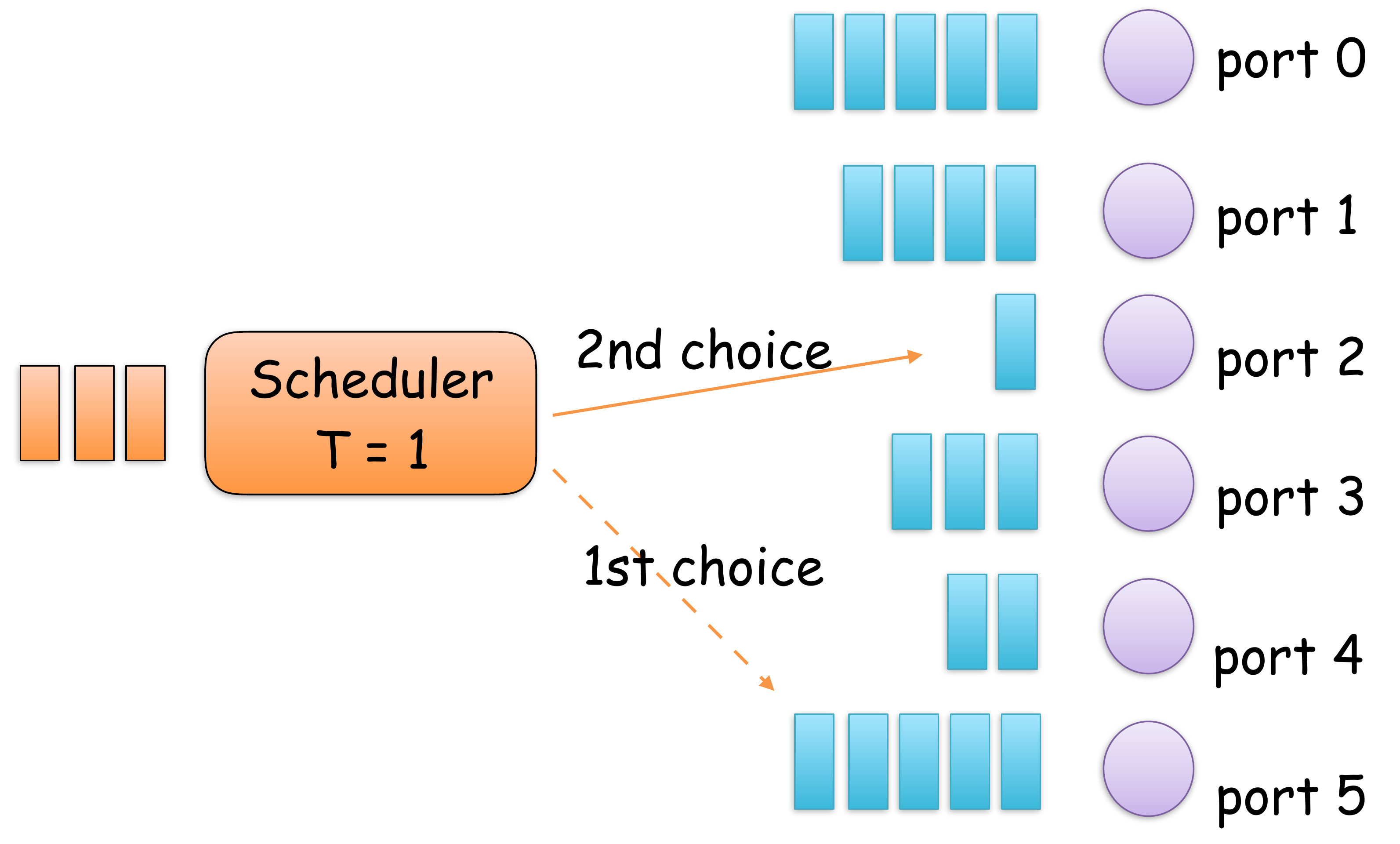} 
} 
\caption{The proposed threshold-based two-choice technique for select one queue out of two.} 
\label{fig:test} 
\end{figure}
We first set D-mod-k paths as default paths for all S-D pairs, then locally and dynamically adjust upward paths of congested S-D pairs through the proposed TTC technique, which is executed independently by each switch without communications and cooperations with other switches.
Because of the fully adaptive upward routing property of fat-tree networks, each switch can adopt the two-choice technique to balance its uplink traffic in the following way.
Upon the arrival of new traffic, a switch first compares the load of the D-mod-k uplink and another randomly selected uplink. If the random choice is less loaded, this switch then redirects the traffic to the randomly selected uplink. 
The traffic is referred to as flows in flow routing and as packets in packet routing. 
The load of a link is measured by the number of flows using this links or by its queue length of packets.

 \begin{algorithm}[!htpb]
 \caption{ Pseudo code of DRB routing scheme.} 
 \KwIn{$h^s, h^d$, T} 
 \KwOut{$P_{up}(h^s, h^d)$} 
 Calculate  $k = \max_i \{ i\mid h^s_i  \neq h^d_i\}$\;
 \If{k ==1}
 { Tag =0\;
 }
 \Else{ Tag = 1\;}
 \If{Tag ==1}
 {
    \For{$i=1; i \le k-1; i$++} 
    { select a random uplink $r_i$\;
      \If{$l(r_i)< l(h^d_i) -T$}{
        $p_{up}^i  =r_i$ \; }
      \Else{
           $p_{up}^i  =h^d_i$\;
           }
        }
 Tag = 0\; 
 }
 \If{Tag == 0}
 {
     \For{$i=k; i \ge 1; i$--} 
    { $p_{dn} = h^d_i$\;
     }
 }
 return $P_{up}(h^s, h^d)$\; 
 \end{algorithm}

Direct application of two-choice, however, usually induces a large amount of redirected traffic, which will significantly increase the level of path collision in the downward stage. To see this, we assume that the load distribution of each uplink of a switch is i.i.d.. Then, we  infer that switch routes traffic to D-mod-k uplinks at each layer with nearly 0.5 probability. Hence, if the traffic needs to transverse $k-1$ layers, then the probability that the traffic remains to be routed to the predetermined D-mod-k path is about  $0.5^{k-1}$.
For example, in a typical three-layer fat-tree network (widely deployed for DCNs), only about 25 percent of traffic is routed to D-mod-k paths. In view of this, the low-level of downward path collision of D-mod-k cannot be guaranteed.
We shall numerically show later that direct incorporation of the two-choice technique into D-mod-k scheme fails to improve the routing performance over D-mod-k.

To address this problem, we propose the TTC technique which introduces a threshold to the two-choice technique, aiming to reduce the redirected traffic to a more reasonable amount.
The description of this technique is as follows.
Set a threshold $T$. Upon the arrival of new traffic destined for $h^d$, a switch at layer $i$ first checks the load of the D-mod-k uplink $h^d_i$, denoted by $l(h^d_i)$, and the load of another randomly selected uplink $r_i$, denoted by $l(r_i)$. The traffic is redirected to  uplink $r_i$ if and only if  $l(r_i) < l(h^d_i) -T$; otherwise, it is routed to the D-mod-k uplink $h^d_i$. 
( Fig. \ref{fig:test} illustrates the TTC technique which is executed independently at every switch in the upward routing stage.)  Algorithm 1 presents the pseudo code of DRB with TTC. 

Intuitively, an appropriate value setting of the threshold is critical for DRB to yield performance improvement. High values usually lead to fewer traffic redirections but are less effective in evenly distributing upward traffic among uplinks. Low values lead to frequent traffic redirections that may high levels of downlink congestion.
We shall numerically show that an appropriate setting of the threshold can yield a significant performance improvement over the case where the threshold is set to be 0.

To gain an in-depth understanding of the TTC technique in load balancing, we conduct  theoretical analysis on the impact of a threshold introduced to the two-choice technique in Section 5, and show that the introduced threshold does not weaken the power of the two-choice technique in load balancing. 

Permutation is an important communication pattern for numerical evaluation and theoretical analysis on the performance of a routing scheme \cite{oblivious_sigmetrics, oblivious_D_Mod_k}.  
In the next section, we design numerical experiments to evaluate the effectiveness of DRB in both flow and packet routing models with permutation flows and permutation packets respectively.
For the flow-routing model,  DRB is applied to dynamic path assignment for each flow.
For the packet-routing model, DRB is applied to allocate uplinks and downlinks for each packet.

\section{Numerical Experiments}
We design the numerical experiments with C++ for two fat-tree networks $\mathcal{F}(3, 24)$ and $\mathcal{F}(4, 12)$.
The experiments compare the performance of DRB with that of VLB, Micro, and D-mod-k under the flow and packet routing settings.
For clarity, we provide a brief explanation on routing schemes: Micro and VLB. 

\textbf{Micro}: Upon arrival of new upward traffic, each switch compares the load of two uplinks that are selected at random, and directs the traffic to the less-loaded uplinks. 

\textbf{VLB}: Valiant's load balancing (also known as Valiant's trick) was first proposed to substantially improve the worst-case performance of permutation routing in Hypercube networks \cite{valiant}. This trick can be naturally applied to fat-tree networks with the same function of improving the worst-case performance of permutation routing because of the its randomness in selecting a path for each pair at each time step. The way to apply VLB routing to a fat-tree network can be as follows. At each time, given a pair $(h^s, h^d)$ with distance $k$, the upward path selected by VLB is 
\begin{displaymath}
P_{up}(h^s, h^d, k) = (r_1,..., r_{k-1}),
\end{displaymath}
where $(r_1, ..., r_{k-1})$ is a randomly picked sequence from the set $\mathbb{Z}^{k-1}_{d}$.

\subsection{Permutation-flow Routing}
To test the effectiveness of the propose scheme in the flow model, we use $c$-permutation flows wherein a host sends exactly $c$ flows to other hosts and receives exactly $c$ flows from other hosts.
Next, we run DRB for flow-path assignment, where each switch sequentially assigns an uplink to every incoming flow using TTC. After all flows are assigned with paths, we then calculate the maximum link load. 
The maximum link load is referred to the maximum number of flows transversing a single link over all links in a fat-tree network including both uplinks and downlinks. The maximum link load is equal to the maximum level of path collision and thus reflects the worst-case performance of a routing scheme.

\begin{figure}
\centering
\subfigure{
\begin{minipage}[b]{0.5\textwidth}
\includegraphics[width=0.9\textwidth]{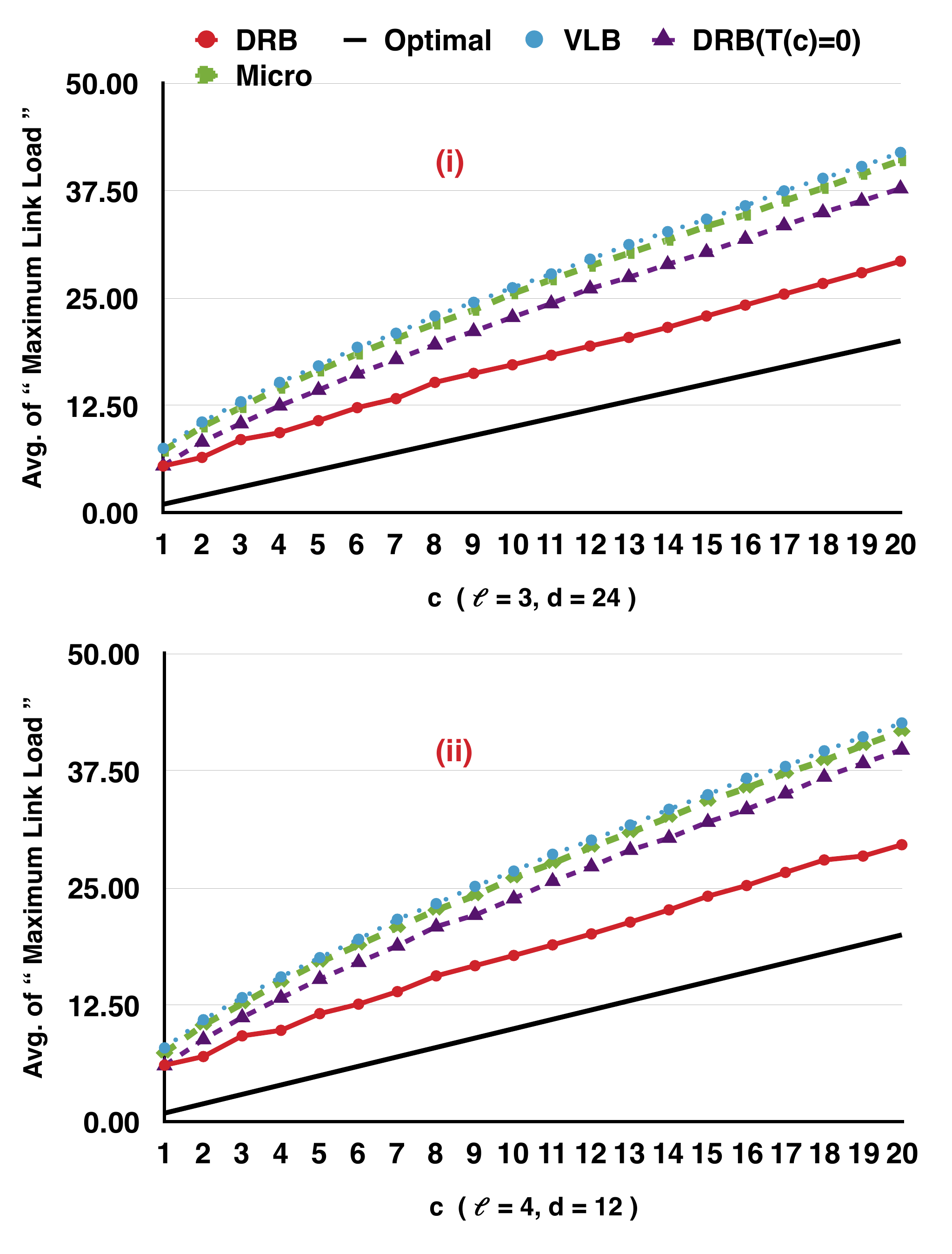}
\end{minipage}
}
\caption{The average of the maximum link load  in  $\mathcal{F}(3, 24)$ and in $\mathcal{F}(4, 12)$ respectively.}
\label{flow_model}
\end{figure}

Specifically, we set the threshold of DRB as follows.
\begin{equation}
\label{c_two}
\begin{cases}
T(c) = \lceil \frac{c}{2}\rceil, & 1\le c <  \ln N,\\
T(c) = \lfloor \frac{\ln N}{2}\rfloor & c > \ln N,
\end{cases}
\end{equation}
where $N$ is the number of hosts in a fat-tree network, which is also equal to the number of uplinks or downlinks at each layer of a fat-tree network.
We run the experiment for 200 times for each value of $c$ ranging from 1 to $20$. We then do average of the maximum link load over 200 experimental results, and present the results in Fig. \ref{flow_model} and in Table 1. It worths mentioning that the variance of the maximum link load over the 200 experiments is small.
Observe that fat-tree is nonblocking to permutation flow, therefore, the maximum link load of $c$-permutation flows under an optimal path allocation should be $c$.
The black line in Fig. \ref{flow_model} reflects the maximum link load under an optimal path allocation. 
 Fig. \ref{flow_model} shows DRB with threshold $(\ref{c_two})$ achieves the best flow-path allocation with a significant performance improvement over Micro, VLB, and  DRB with $T(c) = 0$.

Specifically, we provide a heuristic estimation on the maximum link load as follows.
\begin{equation}
\label{flow_congestion}
c + \frac{\ln \ln N }{\ln 2} + T(c), c\ge 1.
\end{equation}
This estimation is proposed based on two observations.
The first observation is as follows. The derived result in \cite{heavy_load} says that, at the end of sequentially throwing $cN$ balls into $N$  bins with the two-choice technique, the maximum load of a bin is about  
\begin{equation}
\label{max_load}
c + \frac{\ln \ln N}{\ln 2} +\Theta(1),\  c\ge 1
\end{equation}
with high probability, assuming $c\le \Theta(\ln N)$.
The second observation is that, different from the classical two-choice technique, the proposed TTC scheme can not further reduce  load-difference between any two choices whose load difference is already less than or equal to $T(c)$. Therefore, we heuristically adds the extra term $T(c)$ to equation $(\ref{max_load})$ and get the final heuristic estimation. 


\begin{figure}[htbp]
\begin{center}
\includegraphics[width = 3.4 in]{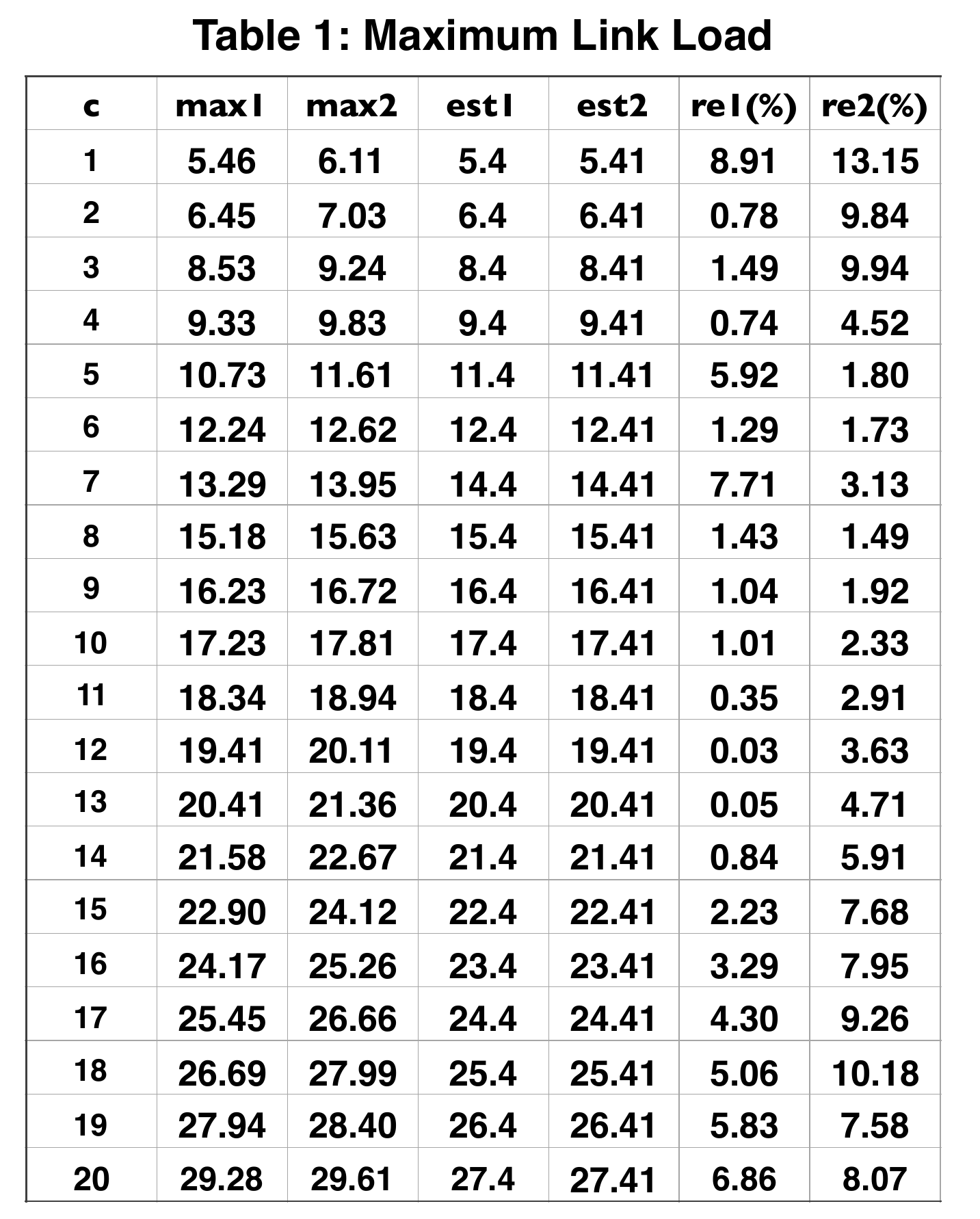}
\end{center}
\end{figure}

\begin{figure}
\centering 
\subfigure[ Average  (i)  and maximum (ii) queue lengths  of up and down links at different layers in $\mathcal{F}(3, 24)$. ] { \label{fig:1} 
\includegraphics[width=1\columnwidth]{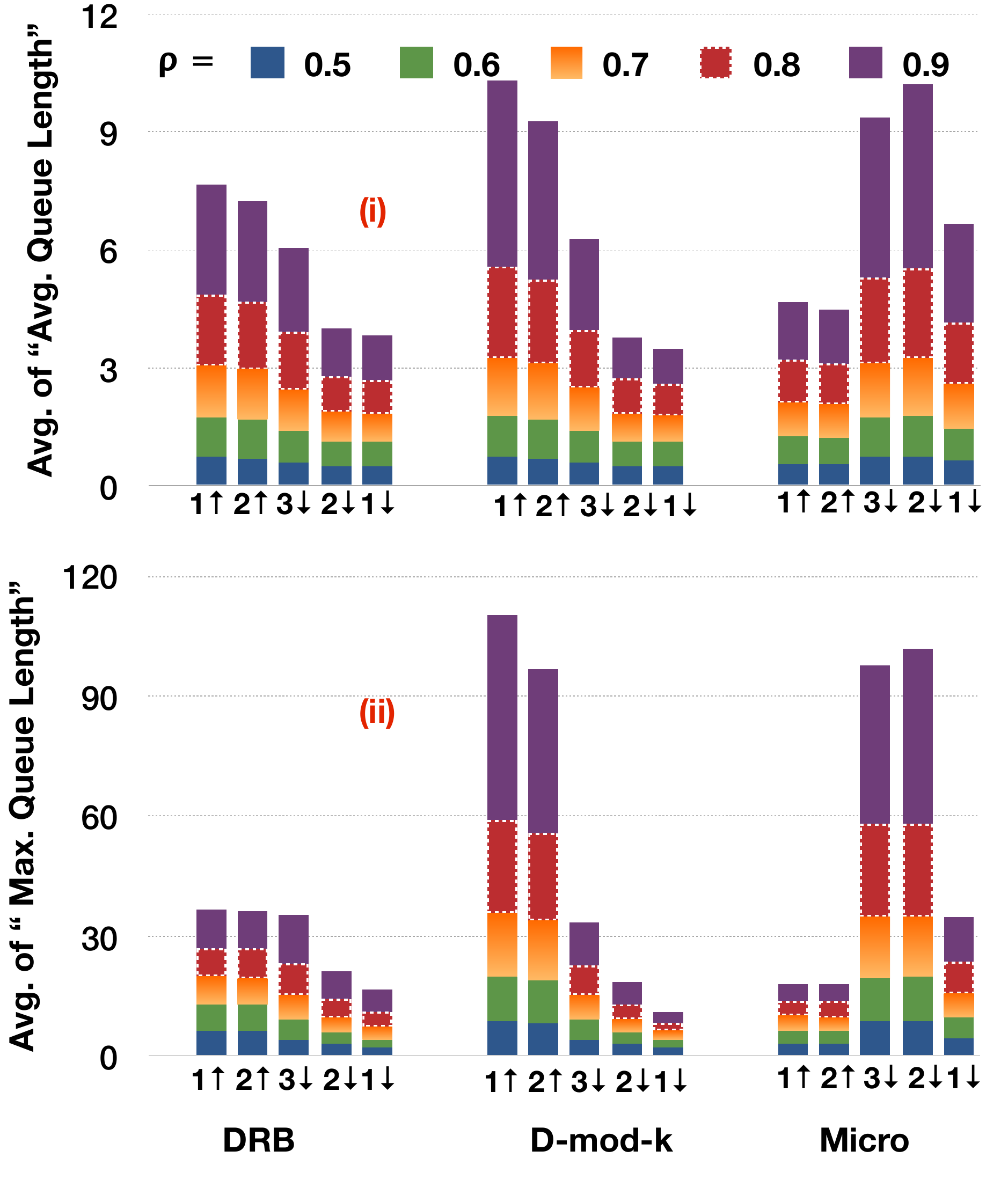} 
} 
\subfigure[Average (iii) and maximum (iv) queue lengths of up and down links at different layers in $\mathcal{F}(4, 12)$.] { \label{fig:2} 
\includegraphics[width=1\columnwidth]{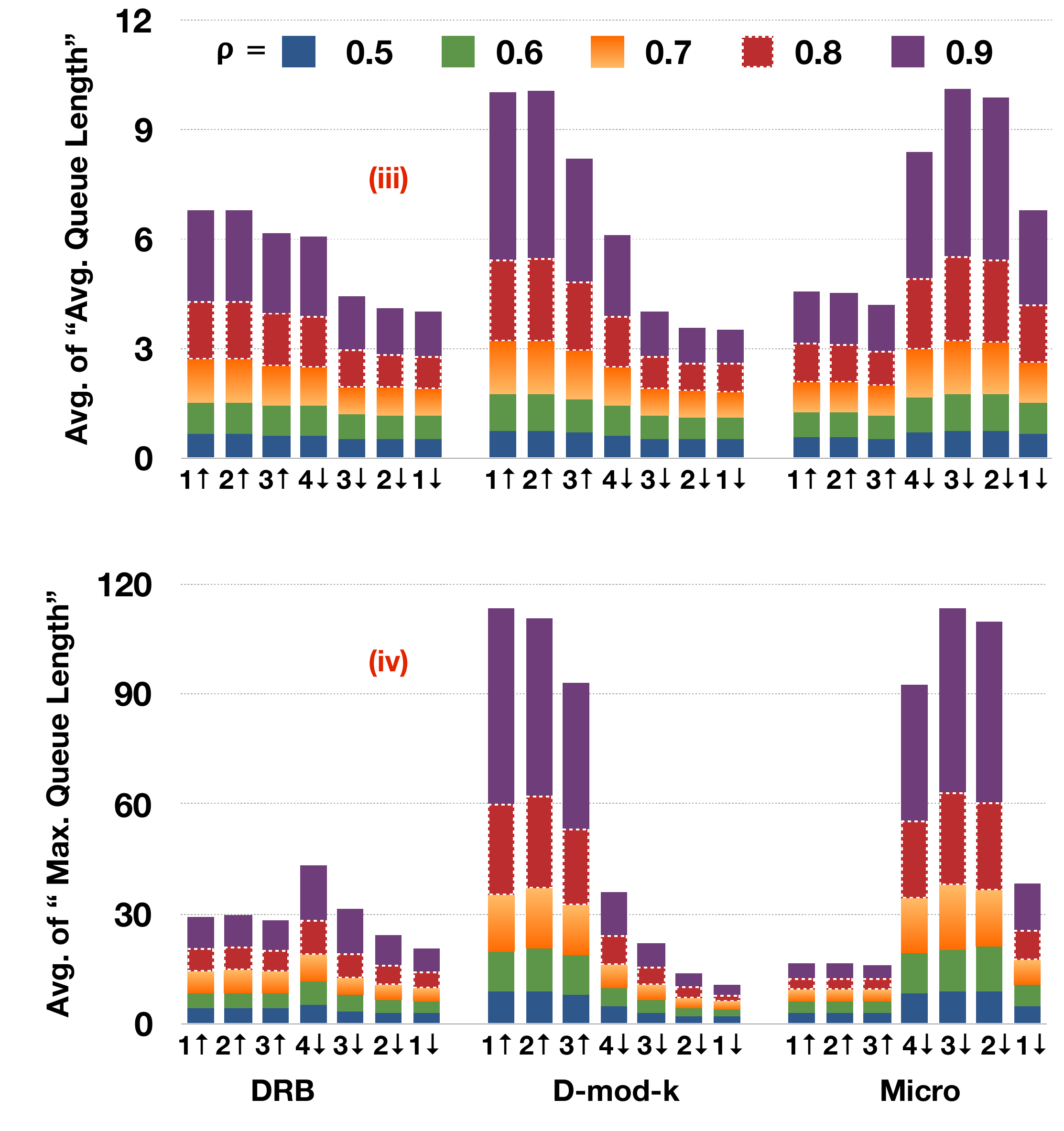} 
} 
\caption{Compare average and maximum queue-lengths under routing schemes: DRB, D-mod-k, and Micro. $k\uparrow$ and $k\downarrow$ indicate uplinks and downlinks of switches at layer $k$s.} 
\label{QueueCompare} 
\end{figure}

\begin{figure*}[!t]
\centering
\subfigure{
\begin{minipage}[b]{1\textwidth}
\includegraphics[width = 6.4 in]{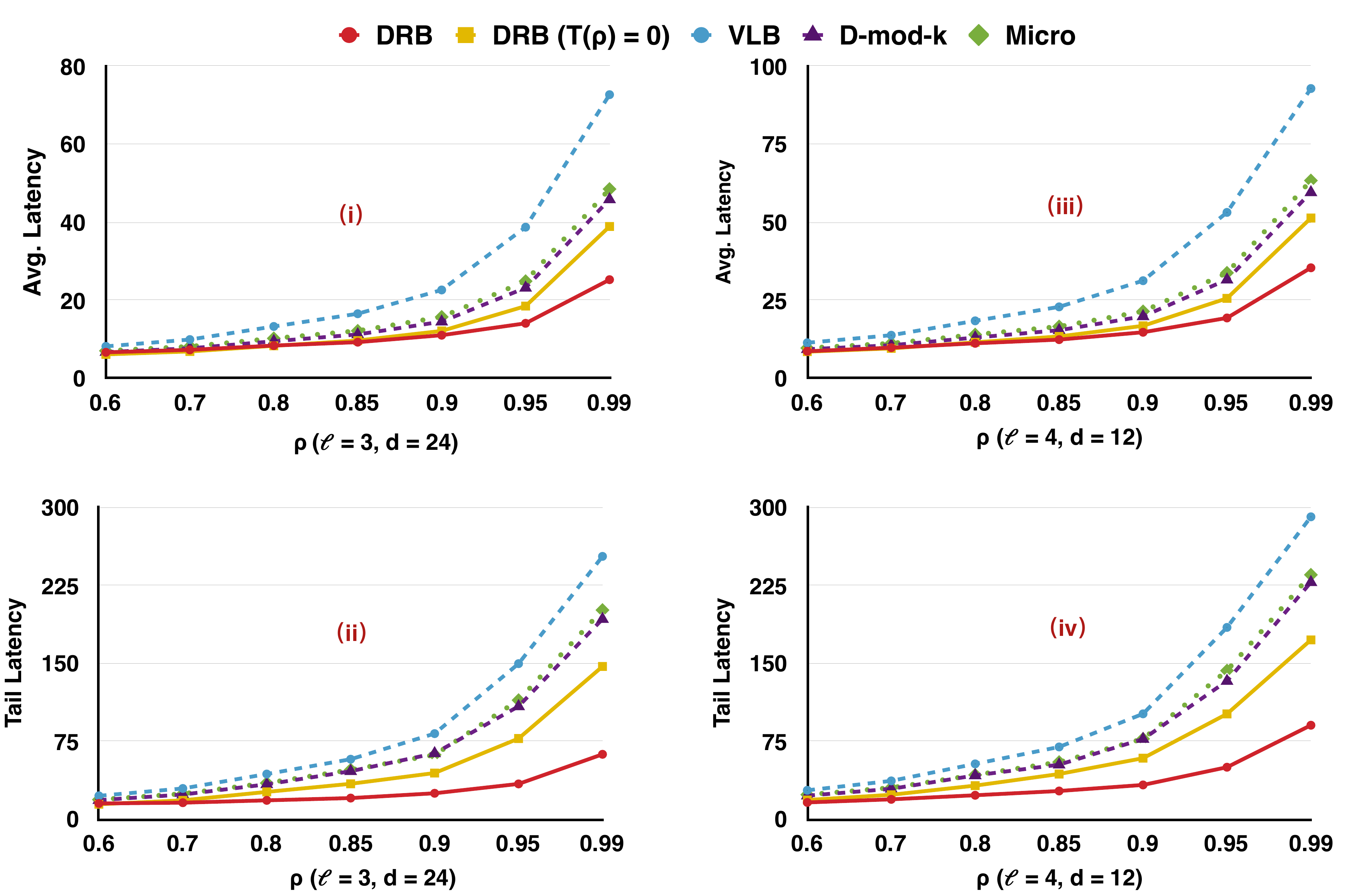}\\
\vspace{-2ex}
\end{minipage}
}
\centering
\caption{DRB achieves the best performance with substantial improvement especially in heavy-traffic scenarios. }
\label{DRB_packet}
\end{figure*}

Table 1 is presented to show the accuracy of the heuristic estimation. Column $max1 $ and $max2$ list the average of the maximum link load of $\mathcal{F}(3, 24)$ and $\mathcal{F}(4,12)$ respectively over 200 repeated experiments. Column $est1$ and $est2$ list the heuristic estimation according to formula (\ref{flow_congestion}) on the maximum link load of $\mathcal{F}(3, 24)$ and $\mathcal{F}(4,12)$ respectively.
Column $re1$ and $re2$ list the relative errors between the simulation and estimation results which are given by $re1 = \frac{| max1- est1 |}{est1}$ and $re2 = \frac{|max2 - est2|}{est2}$ respectively.
Results in$re1$ and $re2$ show that the heuristic estimation is accurate. 
We also observe that $re2$  is slightly larger than that of in $re1$. One explanation for this phenomenon is provided as follows. DRB just locally selects a different uplink from the same switch, whereas  estimation $(\ref{max_load})$ is derived assuming $d$ is large. Therefore, the estimation is more accurate when $d$ is larger.
The results presented in Table 1 and Fig. \ref{flow_model} in turn imply that (1) the number of redirected flows is greatly limited by an appropriate setting of the threshold, and (2) the performance of DRB for flow-routing in fat-tree networks is similar to the performance of the two-choice technique for balanced allocation in the balls-into-bins problem.

\subsection{Permutation-packet Routing}
We continue to test the effectiveness of DRB in the packet model. 
The setup of the experiment is as follows. 
At each time slot, hosts continuously send a set of permutation packets into the network. With probability $\rho$ each host has a packet to send into the network. At most one packet can be transmitted through a single link following the first come first served principle, and it takes one time slot to transmit a packet through a link.
Therefore, $\rho$ must be less than 1 in order to stabilize the system.
Specifically, we set the threshold of DRB as follows.
\begin{equation}
\label{packet_threshold}
T(\rho) = 1-\ln (1-\rho).
\end{equation}
Switches use the TTC technique to send arrival packets to queues of uplinks in their upward routing stage.

%


We run the experiment for $2000$ time steps for each value of $\rho$ ranging from $0.6$ to $0.99$, and do statistics over the last 500 time steps (when we thought the system becomes stable).
We calculate the average and maximum queue length of each layer with results presented in Fig. \ref{QueueCompare}. 
We calculate the average and tail latency of permutation packets that are sent into the network at the same time, and  present results in Fig. \ref{DRB_packet}. 
Tail latency is calculated by the longest time delay of packets that are sent into the network at the same time. 

Results in Fig. \ref{QueueCompare} show that DRB effectively balances both uplink and downlink loads, whereas Micro can balance uplink loads and D-mod-k can only balance  downlink loads only. Moreover, Fig. \ref{DRB_packet} shows that packets routed by the DRB scheme suffer substantially lower average and tail packet latency especially in the heavy-traffic scenarios.
By comparing the performance of DRB with threshold $(\ref{packet_threshold})$  and DRB with $T(\rho ) = 0 $ in Fig.  \ref{DRB_packet}, we can observe that the proposed threshold setting yields a significant performance improvement.

\section{Theoretical Analysis}
In this section, we focus on analyzing the TTC performance in the upward routing stage with the packet-routing model. We first analyze queueing dynamics of uplinks at the 1st layer under general assumptions which are similar to the ones for the supermarket model. 
We then generalize the derived result to other layers assuming that queue-size distributions of different uplinks are independent at the equilibrium state

The assumptions are listed as follows.
Each host send packets into the network following a Poisson process with rate $\lambda$. 
In other words, packets entering into the 1st layer follow the Poisson process of rate $N\lambda$, where $N$ is the number of links (up or down) at each layer of the network.
The destination of each packet is randomly selected among all other hosts in the network.
Each link transmits at most one packet at each time slot  under the first come first served discipline. The transmission time of a packet through a single link follows an exponential distribution with parameter $\mu = 1$. 
Because packets are randomly destined, we claim by symmetry that any uplink of a switch is selected with an identical probability by D-mod-k for a newly-arrival packet. 

\subsection{Queue-size Distribution of the 1st Layer}
We proceed the analysis through the following steps.
We first model queueing dynamics of uplinks at the 1st layer under the TTC technique as a Continuous-Time Markov Chain (CTMC). 
We then show that the CTMC has a unique equilibrium solution by proving that the CTMC is positive recurrent. 
Next, we apply Kurtz's theorem \cite{balance_jump, strongApproximation} to show that the equilibrium queue distribution converges to a fixed-point solution of deterministic differential equations which are derived from the CTMC.
Finally, we prove that the fixed-point solution decays doubly exponentially with respect to queue sizes as $N\to \infty$. 

We denote  by $Q_i(t)$ the queue length of the $i$-th uplink of the 1st layer at time $t$. 
Let $\textbf{Q}(t) = (Q_1(t), ..., Q_N(t))$ denote the corresponding queue size vector.
Let $\mathcal{J} \in \mathbb{Z}^{N}$ be the transition set of $\textbf{Q}(t) $ from one state to another, where $\mathbb{Z}^{N}$ denotes the N-ary Cartesian power of a set $\mathbb{Z} = \{0, 1, ...\}$.
Observe that transition is triggered by either arrivals or departures, so we have
\begin{displaymath}
\mathcal{T} = \{ \pm e_i :  i = 1, ..., N\}, 
\end{displaymath}
where $e_i$ is a standard unit vector. Specifically, $e_i$ indicates the arrival of a packet into the $i$-th queue and $-e_i$ indicates the departure of a packet from the $i$-th queue.  
Clearly, the CTMC is irreducible and nonexplosive.
We further show that the CTMC has a unique equilibrium solution in Theorem \ref{stability_analysis}. Theorem \ref{stability_analysis} is quite similar to Theorem 1 in \cite{slightlyMore}, but the proof for Theorem \ref{stability_analysis} is quite different. 

\begin{theorem}
\label{stability_analysis}
The CTMC ($\textbf{Q}(t)$) is irreducible, nonexplosive, and positive recurrent. Furthermore, at the equilibrium state, it satisfies 
\begin{displaymath}
E\left[ \frac{1}{N} \sum_{i=1}Q_i(t)\right] < c,
\end{displaymath}
for any $N$ with $c =  \frac{1}{1-\lambda}$. 
\end{theorem}
The proof for this theorem is presented in Appendix C. 

\subsection{Deterministic Representation of the Equilibrium-state CTMC}
To investigate the tail distribution of queue sizes, we use the following notation to represent the state of queues at the 1st layer. 
Let $n_i(t)$ denote the number of queues with $i$ packets at time $t$.
Then $m_i(t) = \sum_{k =i}^\infty n_k(t)$ represents the number of queues with at least $i$ packets at time $t$.
Correspondingly, $p_i(t) = \frac{n_i(t)}{N}$ denotes the fraction of queues with $i$ packets,
and $s_i(t) = \frac{m_i(t)}{N} = \sum_{j=i}^\infty p_j(t)$ denotes the fraction of queues with at least $i$ packets.
It worths noting that $p_i(t) = s_i(t) - s_{i+1}(t)$.
We then represent the state of queues at the same layer at time $t$ by vector 
\begin{math}
M(t) = (m_0(t), m_1(t), ...).
\end{math}
Consider $i$ could be infinitely large, the dimension of $M(t)$ is infinite. 
We re-define transition set by $\textbf{L} = \{l_0, \pm l_i \mid i= 1, 2, ..\}$, where $l_i$ indicates a packet arriving at one of the queues of size $i-1$, and $-l_i$ implies a packet departing from one of the queues with size $i$.
Note that if a packet is sent to a queue of size $i-1$, then the value of $m_i(t)$ increases by one and the value of other $m_j(t)$ $\forall j\neq i$ remains the same.
We then claim that the state-dependent transition rate, $\beta_{l}(M(t))$ , is given by
\begin{displaymath}
\begin{split}
&\beta_{\ell_i}(M(t)) = N \lambda p_{i-1}(t)(s_{i-T-1}(t) +s_{i+T}(t)), i\ge 1, \\ 
&\beta_{-\ell_i}(M(t)) = Np_i(t), i\ge 1.
\end{split}
\end{displaymath}

The claim for the departure rates is obvious. 
In what follows, we only need to verify the above claim for arrival rates. 
First, we need to calculate the probability that a packet is sent to a queue of size $i-1$. If the probability is $p$, then the  transition rate of $l_i$ is $N\lambda p$. 
To calculate the probability, we consider two disjoint events.  
The first event is that the first queue has size $i-1$ and the second queue has a size of at least $i-T-1$. 
The occurrence probability of this event is $\frac{1}{N}p_{i-1}s_{i-T-1}$.
The second event is that the first queue has a size of at least $i+T$ and the second queue has size $i-1$.
The probability of this event is $p_{i-1}s_{i+T}$.
In conclusion, the arrival rate of any queue of size $i-1$ is $N \lambda p_{i-1}(s_{i-T-1} + s_{i+T})$.

Then, $M(t)$ evolves over time $t$ as follows.  
\begin{displaymath}
M(t)  = M(0) + \sum_{l\in \textbf{L} } l N_{l}(t),
\end{displaymath}
where $N_l(t)$ counts the number of a specific transition $l$ up to time $t$, and can be formulated as follows.
\begin{displaymath}
N_{l}(t) = Y_{l}(\int_{0}^{t}\beta_{l} (M(u)) du), 
\end{displaymath}
where $Y_{\textit{l}}$ indicates an independent unit Poisson process.
It then follows that
\begin{displaymath}
M(t) = M(0) +\sum_{l\in \textbf{L}} lY_{l} (\int_{0}^{t}\beta_{l}(M(u)) du ).
\end{displaymath}
Denote $S(t) = (s_0(t), s_1(t), ...,)$ (a vector with an infinite dimension). Divide at both sides of the above equation by $N$, we have 
\begin{equation}
\small
\label{s_equation}
\begin{split}
S(t) &= S(0) + \sum_{l\in \textbf{L}} \frac{l}{N} Y_{l}  (N\int_{0}^{t}\tilde{\beta}_{l} (S(u)) du)\\
& = S(0) + \sum_{l\in \textbf{L}}  \frac{l}{N}\tilde{Y}_{l}  (N\int_{0}^{t}\tilde{\beta}_{l} (S(u))  du) +\int_{0}^t F(S(u))du,
\end{split} 
\end{equation}
where $\tilde{Y}(x) =  Y(x) -  x $ is a Poisson process centered at its expectation,  $\tilde{\beta}_{l}(S(t)) = \frac{1}{N}\beta_{l}(M(t))$, and
\begin{displaymath}
F(S(u)) = \sum_{l\in \textbf{L}} l \tilde{\beta}_{l} (S_N(u)) du),
\end{displaymath}
where 
\begin{displaymath}
\tilde{\beta}_{l}(S(t)) =
\begin{cases}
\lambda p_{i-1}(t)(s_{i-T-1} (t) + s_{i+T}(t)), & l =  l_i, i \ge 1,  \\
p_i(t), & l = -l_i, i \ge 1.
\end{cases}
\end{displaymath}
The initial condition $s_0 =1$ is obvious, whereas the initial conditions $s_{-T} = ... = s_{-1} = 1$ are appended to achieve a uniform formula expression for items $i$, $i\le T$.
As $N\to \infty$, the value of the centered Poisson process shall go to 0 by the law of large numbers.
According to the Kurtz's theorem \cite{strongApproximation, limit_kurtz}, $\lim_{N\to \infty } S(t)  = s(t)$ under mild conditions with $s(t)$ defined as below.
\begin{equation}
\label{deterministics}
s(t)  = s(0) +\int_{0}^t F(s(u))du, 
\end{equation}
where $s(0) = \lim_{N\to \infty} S(0)$.
Some details on applying the Kurtz's theorem is provided in Appendix D.

Previously, we have shown that the CTMC has an equilibrium state. 
 At the equilibrium state, we have $\frac{ds(t)}{dt} = F(s(t)) = 0$. 
In other words, the equilibrium solution can be well approximated by the fixed point solution of equation $(\ref{deterministics})$ when $N$ is large. 
We denote by $(s_i)$ (without $t$ parameter) the fixed-point solution of $(\ref{deterministics})$. 
Then we have
\begin{equation}
\label{fluid-limit}
\begin{cases}
F(s) =\lambda p_{i-1} (s_{i-T-1} +s_{i+T}) -(s_i -s_{i+1}) = 0,  i \ge 1\\
s_{-T}  =  ...  =  s_{0} = 1.
\end{cases}
\end{equation}

\subsection{Doubly Exponential Decay of Queue Sizes}
\begin{defn}
We say sequence $(s_i)$ decays doubly exponentially if there exist positive constants $\lambda <1$, $c$, $\alpha >1$, and $n$, such that for all $i\ge n$, it has $s_i \le c\lambda^{\alpha^i}$. 
\end{defn}

\begin{theorem}
The fixed-point solution of equation (\ref{deterministics}), denoted  $(s_i)$, decays doubly exponentially.
\end{theorem}

\begin{proof}
If $T= 0$, it becomes the classical two-choice problem  \cite{TwoChoices} with equations in (\ref{fluid-limit}) being
\begin{displaymath}
s_{i} -s_{i+1}  = \lambda s^2_{i-1} - \lambda s^2_{i}.
\end{displaymath}
Then we get
\begin{displaymath}
s_i  =\sum_{j = i}^{\infty} (s_i -s_{i+1})  = \lambda \sum_{j=i}^\infty (s^2_{i-1} - s^2_{i}) = \lambda s^2_{i-1}.
\end{displaymath}
Since $s_0= 1$, one can prove by inductive reasoning that $s_i = \lambda^{2^i-1}$ which decays doubly exponentially.
We next focus on the proof for $T\ge 1$, which is much more complicated.

By rearranging equation (\ref{fluid-limit}), we get
\begin{displaymath}
\small
\begin{split}
&p_i = \lambda p_{i-1} [(s_{i-T-1} - s_{i-1})  +(s_{i-1} +s_i) - (s_i -s_{i+T})] \\
& =  \lambda p_{i-1}(s_{i-1} +s_i) + \lambda  p_{i-1}  (s_{i-T-1}  -s_{i-1}) - \lambda p_{i-1}(s_i -s_{i+T}).
\end{split}
\end{displaymath}
To facilitate the proof, we define 
\begin{equation}
\nonumber
\begin{split}
&\delta_i =\lambda p_{i-1} (s_{i-1}  + s_i)  = \lambda s^2_{i-1} -\lambda s^2_{i}, \\
&\gamma_i =  \lambda  p_{i-1}  (s_{i-T-1}  -s_{i-1}) = \lambda p_{i-1}(p_{i-T-1} + ... + p_{i-2}),\\
&\eta_i = \lambda p_{i-1}(s_i -s_{i+T}) = \lambda p_{i-1}(p_i + ...+ p_{i+T-1}).
\end{split}
\end{equation}
By rearranging items in $ \eta_i $, we have 
\begin{equation}
\begin{split}
\frac{1}{\lambda}\sum_{j= i}^\infty  \eta_j  &=  p_i p_{i-1} + p_{i+1}(p_{i-1}+ p_i ) + ... \\
&+ p_{i+T-1}(p_{i-1} + ... + p_{i+T-2} )\\
&+\sum_{k=i+T}^\infty p_k(p_{k -T} + ... +  p_{k-1})\\
&= p_i p_{i-1} + p_{i+1}(p_{i-1}+ p_i ) + ... \\
& + p_{i+T-1}(p_{i-1}+...+ p_{i+T-2} ) + \sum_{k = i+T}^\infty \gamma_{k+1} .\\
\end{split}
\end{equation}
It follows that
\begin{equation}
\small
\label{elimination} 
\begin{split}
&\frac{1}{\lambda}( \sum_{j=i}^\infty \gamma_j -\sum_{j=i}^\infty \eta_j)  \\
& = \sum_{j= i}^{i+T-1}p_{j-1}(p_{j-T-1} + ... + p_{j-2}) + \sum_{k=i+T}^{\infty} \gamma_k\\
&- [p_i p_{i-1} + p_{i+1}(p_{i-1}+ p_i ) + ...+ p_{i+T-1}(p_{i-1} +...+ p_{i+T-2} )]\\
& -\sum_{k= i+T}^\infty \gamma_{k+1}\\
&= \sum_{j= i}^{i+T-1}p_{j-1}(p_{j-T-1} + ... + p_{j-2})  \\
& -[p_i p_{i-1} + p_{i+1}(p_{i-1}+ p_i ) + ...+ p_{i+T-2}(p_{i-1}+...+ p_{i+T-3} )]\\
&=\sum_{j= i}^{i+T-1} p_{j-1}( s_{j-T}- s_{i-1}).
\end{split}
\end{equation}
Then we have
\begin{equation}
\label{label_s}
\begin{split}
&s_i \triangleq \sum_{j = i}^\infty p_j  =\sum_{j =i}^\infty(\delta_j + \gamma_j - \eta_j) \\
&= \lambda s^2_{i-1}+  \lambda \sum_{j =i}^{i+T-1}  p_{j-1} (s_{j-T} -s_{i-1}).
\end{split}
\end{equation}
The first item in the last equality is derived from equation
\begin{math}
\sum_{j=i}^\infty \delta_j = \lambda s^2_{i-1},
\end{math}
and the latter item is derived from equation (\ref{elimination}).

Reformulate $s_i$ as follows.
\begin{displaymath}
s_i  = \hat{s}_i+  \lambda \sum_{j =i+1}^{i+T-1}  p_{j-1} (s_{j-T} -s_{i-1}). 
\end{displaymath}
where $\hat{s}_i =\lambda s^2_{i-1} +\lambda p_{i-1}( s_{i+T-1} -s_{i-1})$. 
We then claim that
\begin{equation}
\label{claim_s}
s_i \le \lambda s_{i-1}s_{i-1-T}.
\end{equation}
Proof of the claim. 
First of all, we have
\begin{equation}
\nonumber
\begin{split}
\hat{s}_i
& = \lambda s^2_{i-1} + \lambda (s_{i-1} -s_{i})(s_{i-T-1} -s_{i-1})\\
& = \lambda s^2_{i-1} +\lambda s_{i-1}s_{i-T-1} -\lambda s^2_{i-1} -\lambda s_i(s_{i-T-1} -s_{i-1})\\
& = \lambda s_{i-1}s_{i-T-1} - \lambda s_i (s_{i-T-1} -s_{i-1}). \\
\end{split}
\end{equation}
On the other hand, we have 
\begin{displaymath}
\small
\begin{split}
&- s_i (s_{i-T-1} -s_{i-1})  + \sum_{j=i+1}^{i+T-1} p_{j-1}(s_{j-T} -s_{i-1} ) \\
&\le  - s_i (s_{i-T-1} -s_{i-1}) +\sum_{j = i+1}^{i+T-1} p_{j-1}(s_{i-T} - s_{i-1} )\\
& =  (s_{i-T-1} -s_{i-1} ) (\sum_{j=i+1}^{i+T-1} p_{j-1} -s_i) \le 0.\\
\end{split}
\end{displaymath}
The claim is thus validated. 

According to the initial condition, $s_{j} = 1, j\in [-T, 0]$, we have
\begin{displaymath}
s_i \quad
\begin{cases}
= 1, & i < 1\\
\le  \lambda s_{i-1}s_{i-T-1}, & i \ge 1.
\end{cases}
\end{displaymath}
To derive an upper bound of the tail distribution of $(s_i)$, we define a sequence $(z_i)$ as follows.
\begin{equation}
z_{i} = 
\begin{cases}
1, & i\in [-T, 0],\\
\lambda  z_{i-1} z_{i-T-1},& i\ge 1.
\end{cases}
\end{equation}
It is easy to see $z_i\ge s_i$ and $z_i  = \lambda^i$ for $i\in [0, T+1]$. 
Moreover, we have
\begin{displaymath}
\ln_\lambda z_i  =  1+ \ln_{\lambda} z_{i-1} +\ln_{\lambda} z_{i-1-T}, i \ge 1.
\end{displaymath}
Define 
\begin{equation}
\label{transformation}
\ln_{\lambda} z_i  +1 =  g_i.
\end{equation}
We get
\begin{equation}
\label{recursive}
g_k  = k+1,  k \in [0, T+1]; \ g_i = g_{i-1} + g_{i-1-T}\  \forall i\ge 1.
\end{equation}


Fix $T$ with $T\ge 1$. Clearly, there exists an $\alpha >1$ such that $\alpha^T \le \frac{1}{\alpha-1}$. That is 
\begin{equation}
\label{alpha_setting}
\alpha^T +1\ge \alpha^{T+1}.
\end{equation}
If $g_{j} \ge c\alpha^{j} $  for all $j <i$,  according to equation (\ref{alpha_setting}),  we have 
\begin{displaymath}
g_i  = g_{i-1} + g_{i-T-1}\ge c \alpha^{i-1} + c\alpha^{i-1-T} \ge c\alpha^i.
\end{displaymath}
In particular, $c$ can be set as follows:
\begin{displaymath}
c =  \min_{i\in [0, T+1]} \frac{i+1}{\alpha^i},
\end{displaymath}
which implies that  
\begin{equation}
\label{recursive2}
g_i \ge c\alpha^i, \forall i\in[0, T+1].
\end{equation}
For any $T\ge 1$, we conclude that there exists a positive constant $c$ and an $\alpha >1$ such that $g_i \ge c\alpha^i $ for all $i\ge 1$.
In other words, we have 
\begin{displaymath}
s_i \le z_i \le  \lambda^{c\alpha^i-1}, \forall i\ge 1.
\end{displaymath}
We thus finish the proof for the theorem. 
\end{proof}

\subsection{Queue-size Distributions of Higher Layers}
To estimate the queue-size distribution of other layers, we assume that the equilibrium distribution of each queue at the 1st layer is identical and asymptotically independent.
Intuitively, we infer by symmetry that each queue has an identical distribution at the equilibrium state. 
Furthermore, Bramson et al. \cite{mean_field_2} validated the equilibrium distribution of each queue is independent to each other when considering the supermarket model as long as that the arrival rate is small and the service distribution has 1st and 2nd moments. 
Check the size of any queue of the 1st layer at any time after the CTMC of the 1st layer arrives at the equilibrium state. A queue is found to have $i$ packets with probability $p_i$ and to have at least $i$ packets with probability $s_i$.

\begin{figure}[ht]
\begin{center}
\includegraphics[width = 3 in]{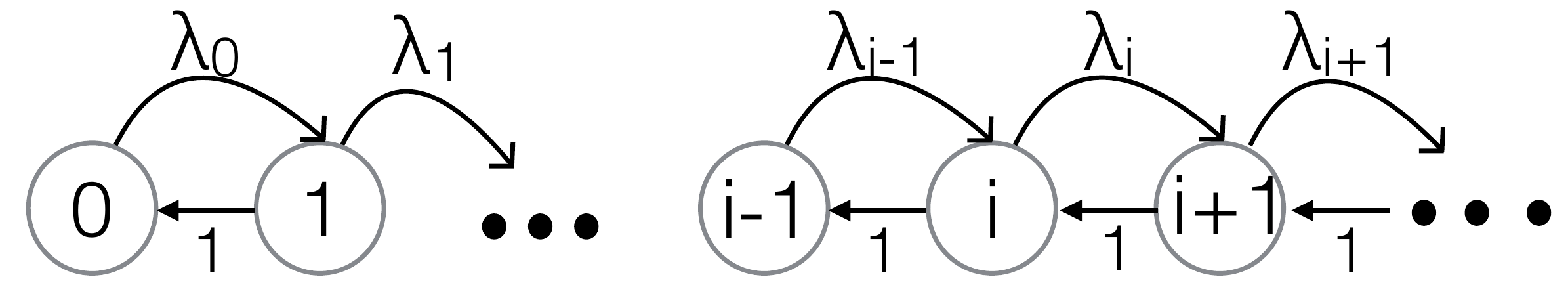}
\caption{The CTMC of a single queue.  $i$ indicates that the queue is at state $i$ with $i$ packets.}
\label{Markov_Chain}
\end{center}
\end{figure}

In what follows, we consider a fixed queue of the 1st layer and aim to show that the CTMC of this queue (see Fig. \ref{Markov_Chain} ) is time-reversible at the equilibrium state.
Obviously, arrival rate $\lambda_i$ in Fig. \ref{Markov_Chain} in is state-dependent because of load comparison brought by TTC. 
Assume the current state of this queue is $i$. 
Then, we claim that the probability that a newly arrival packet is sent to this queue is $\frac{1}{m} (s_{i-T} + s_{i+T+1})$, where $m$ is the number of available output queues in a switch.
$m =d$ if the switch is an intermediate switch; $m=2d$ if the switch is a core switch.
To prove the claim, we only need to consider two disjoint events. The first event is this queue is selected as the first choice while another queue selected as the second choice has at least $i-T$ packets. 
The second event is this queue is selected as the second choice while another queue selected as the first choice has at least $i+T+1$ packets.
Note that the probability that this queue is selected as the first and second choice is $\frac{1}{m}$ and  $\frac{m-1}{m}\frac{1}{m-1} = \frac{1}{m}$ respectively.
As a result, the occurrence probability of the first and the second events is $\frac{1}{m}s_{i-T}$ and $\frac{1}{m}  s_{i+T+1}$ respectively.

Observe that arrival rate of a single switch is $m\lambda$. Therefore, for any $i\ge 1$, we have
 \begin{equation}
 \label{transition_rate}
\lambda_i = (m \lambda)*[ \frac{1}{m}(s_{i-T} + s_{i+T+1})] = \lambda (s_{i-T} + s_{i+T+1}).
\end{equation} 
According to the previous result in equation (\ref{fluid-limit}), we can see the local-balance equations of a single-queues CTMC is satisfied. That is, we have 
 \begin{equation}
\label{local_balance}
p_{i+1} = \lambda_i p_i,
\end{equation}
which implies that the single-queue CTMC is time-reversible.
We thus conclude that the departure process of a single queue has the same distribution as the arrival process, namely, a Poisson process.
Moreover, since \begin{math}
\sum_{i=0}^\infty \lambda_i p_i = \sum_{i=0}^\infty  p_{i+1} = \sum_{i=1}^\infty p_i = \lambda,
\end{math}  the expected departure rate of a single queue is $\lambda$.

By the properties of the time-reversible CTMC and by inductive reasoning, we further infer that packet arrivals of each layer follow a Poisson process with approximate rate $ N\lambda$. As a result, the equilibrium queue-size distribution of each layer should be independent and identical. In particular, the numerical results in Fig. \ref{QueueCompare} also show that average and maximum queue lengths of uplinks at different layers are roughly the same.

\section{Conclusion}
This paper presents a simple and effective  load-balanced routing scheme called DRB for fat-tree networks, which incorporates the randomized load-balancing technique called TTC into the deterministic D-mod-k routing scheme.
TTC uses a threshold to effectively reduce traffic redirection operation, which contributes  to evenly distributing traffic among uplinks and downlinks.
The experimental results show that DRB succeeds to achieve low-levels of path collision in the flow model and low average and tail latency in the packet model. 
Theoretical results show that the performance of TTC in balanced allocation is similar to that of the two-choice technique.

\bibliographystyle{IEEEtran}
\bibliography{sigproc} 

\appendices
\section{Proof of the Integer-fixing Downward Routing}
Consider an S-D pair $(h^s, h^d)$ with distance $k$.
According to the host-switch connection definition, the source host connects directly to the switch with label code  $(h^s_2, ..., h^s_\ell)$; the destination host connects directly to the switch with label code $(h^d_2, ..., h^d_\ell)$.
According to the switch-switch connection definition in equation (\ref{switch_connection}), we can infer the label codes of switches transversed by the upward and downward paths at each layer as follows.
\begin{enumerate}[(i)]
\item The upward path $ (p_{up}^1, ..., p_{up}^{k-1})$ transverses a switch at layer $i$ ($ i \le k-1$) whose label code is 
\begin{equation}
\label{up_switch}
(p_{up}^1,..., p_{up}^{i-1}, h^s_{i+1}, .., h^s_{\ell}),
\end{equation} 
\item The label code of the \textbf{transition} switch at layer $k$ is
\begin{equation}
\label{top_switch}
(p_{up}^1,..., p_{up}^{k-1}, h^s_{k+1}, .., h^s_{\ell}).
\end{equation}
\item The  downward path $(p_{dn}^k, ..., p_{dn}^1)$ transverses a switch at layer $j$ $(1\le j \le k)$ whose label code is
\begin{equation}
\label{down_switch}
(p_{up}^1, ..., p_{up}^{j-1},  p_{dn}^{j+1}, ..., p_{dn}^{k}, h^s_{k+1}, ..., h^s_{\ell}).\\
\end{equation}
\end{enumerate}
Obviously, the switch-label for $j=1$ is $(p_{dn}^2, ..., p_{dn}^k, h^{s}_{k+1}, ..., h^s_{\ell})$.
In what follows, it has
\begin{equation}
\label{edge_switch}
(p_{dn}^2, ..., p_{dn}^k, h^{s}_{k+1},..., h^{s}_\ell ) = (h^d_2, ...., h^d_\ell ), \ p_{dn}^1 = h^d_1.
\end{equation}
Therefore, we can explicitly represent the downward path as follows.
\begin{equation}
\label{integer_fixing}
P_{dn}(h^s, h^d, k)= (h^d_{k}, ..., h^d_1).
\end{equation}
Moreover, we re-write down (\ref{down_switch}) as follows.
\begin{equation}
\label{down_2}
(p_{up}^1,..., p_{up}^{j-1}, h^d_{j+1}, ..., h^d_{\ell}).
\end{equation}
Thus, we finish the proof.

\section{Proof of D-mod-k}
To prove this theorem, we only need to show that D-mod-k paths of any two pairs with different destinations do not share any common downlinks.
Consider any two S-D pairs $(h^s, h^d)$ and $(\hat{h}^s, \hat{h}^d)$ with distance $k$ and $\hat{k}$ respectively. 
The corresponding D-mod-k upward paths are $(h^d_1, ..., h^d_{k-1})$ and $(\hat{h}^d_1, ..., \hat{h}^d_{k-1})$.
According to equation (\ref{down_switch}), the label codes of switches at layer $m$ that are transversed by the downward paths of these two S-D pairs can be expressed follows. 
\begin{equation}
\label{up_port_2}
\begin{cases}
(h^d_1 , ..., h^d_{m-1}, h^d_{m+1} , ..., h^d_\ell), &m \le k\\
( \hat{h}^d_1, ..., \hat{h}^d_{m-1}, \hat{h}^d_{m +1} , ..., \hat{h}^d_\ell), &m \le \hat{k}.
\end{cases}
\end{equation}
Suppose the two downward paths collide at layer $m$ ($m\le \min_{k, \hat{k}}$), which means that the two paths transverse the same down port of a switch at layer $m$. Then, we have 
\begin{displaymath}
\begin{cases}
(h^d_1 , ..., h^d_{m-1}, h^d_{m+1} , ..., h^d_\ell) = ( \hat{h}^d_1, ..., \hat{h}^d_{m-1}, \hat{h}^d_{m +1} , ..., \hat{h}^d_\ell),\\
h^d_m = \hat{h}^d_m.
\end{cases}
\end{displaymath}
The former equation implies that the two downward paths pass through the same switch at layer $m$, the latter equation implies that they transverse the same down-port of this switch. A direct implication follows that the two pairs must have the same destination. 
Because this derivation applies to any $m$ with $m\le \min\{k, \hat{k}\}$, we conclude that pairs with different destinations must have different downlinks. Thus, we finish the proof of this theorem.


\section{Stability Analysis}
\begin{proof}
Clearly, the CTMC  is irreducible and nonexplosive.
In what follows, we aim to show the CTMC is also positive recurrent.
We define a Lyapunov function $V(\textbf{Q}(t)) = \sum_{i=1}^{N} Q^2_i(t)$, where $\textbf{Q}(t)\in \mathbb{Z}^N$ indicates CTMC state at time $t$.  
The domain of $\textbf{Q}(t+1) $ is $ \{\textbf{Q}(t) + \textbf{t} \mid \textbf{t} \in \mathcal{T}\}$.
Denote by $q_{e_i}$ and $q_{-e_i}$ the arrival rate and departure rate of the $i$-th queue respectively. 
Clearly, 
\begin{math}
q_{-e_i}(t)  = \mu = 1.
\end{math}
Moreover, 
\begin{equation}
\label{departure_equ}
q_{-e_i} (t)[ V(\textbf{Q}(t)-e_i) - V(\textbf{Q}(t))]^+  \le  -2Q_i(t) + 1.
\end{equation}
Assume a newly arrival packet joins queue $i$ with probability $p_{e_i}(t)$, then the expected  Lyapunov drift is given by\begin{equation}
\begin{split}
& \sum_{i=1}^N  p_{e_i}(t)[ V(\textbf{Q}(t)+e_i) - V(\textbf{Q}(t))] =  \sum_{i=1}^N p_{e_i}(t) (2Q_i +1 ) \\
\end{split}
\end{equation}
To derive an upper bound of the expected Lyapunov drift, we compare the TTC  scheme with the single-choice scheme provided that the current state is $\textbf{Q}(t)$. 
The single-choice scheme always sends a newly arrival packet to a queue selected at random. The corresponding expected Lyapunov drift caused by an arrived packet is $\frac{1}{N}\sum_{i=1}^N  (2Q_i +1 )$.
In contrast,  the TTC technique sends the packet to a randomly selected queue if and only if the length of the randomly selected queue is at least $T$ less than a queue determined by D-mod-k.
Therefore, the expected Lyapunov drift under the TTC technique must be no larger than $\frac{1}{N}\sum_{i}^N (2Q_i +1)$.
In conclusion, we have 
\begin{displaymath}
\frac{1}{N} \sum_{i=1}^N  p_{e_i}(t)[ V(\textbf{Q}(t)+e_i) - V(\textbf{Q}(t))] \le \frac{1}{N}\sum_{i=1}^N (2Q_i +1).
\end{displaymath}
Recall the arrival rate of packets is $N\lambda$, then 
$q_{e_i }(t) = N\lambda p_{e_i}(t)$.
As a result, we have
\begin{equation}
\label{arrival_equ}
 \sum_{i=1}^N  q_{e_i}(t)[ V(\textbf{Q}(t)+e_i) - V(\textbf{Q}(t))] \le \lambda \sum_{i=1}^N (2Q_i +1).
\end{equation}

Combining Lyapunov drift caused by both the arrival and departure of packets at time $t$, we have
\begin{displaymath}
\small
\begin{split}
&\sum_{\textbf{t}\in \mathcal{T}} q_{\textbf{t}} [V(\textbf{Q}(t)+\textbf{t}) -V(\textbf{Q}(t))] \le 2(\lambda -1)\sum_{i=1}^N Q_i(t)  + (\lambda +1)N.\\
\end{split}
\end{displaymath}
If $\sum_{i=1}^N Q_i(t) > \frac{1+\lambda)N}{2(1-\lambda)}$, then the above sum is less than a negative value. 
By the Foster-Lyapunov theorem (see theorem of book  \cite{srikant2013communication}), we infer that the CTMC is positive recurrent. 
At the equilibrium state, we have 
\begin{displaymath}
\small
\begin{split}
0 &= E\left[\sum_{\textbf{t}\in \mathcal{T}} q_{\textbf{t}} (V(\textbf{Q}(t)+\textbf{t}) -V(\textbf{Q}(t))) \right] \\
&\le 2(\lambda -1)E[\sum_{i=1}^NQ_i(t)]  + (\lambda +1)N, 
\end{split}
\end{displaymath}
which implies that 
\begin{displaymath}
E[\frac{1}{N}\sum_{i=1}^NQ_i(t)] \le \frac{1+\lambda }{2-2\lambda } < \frac{1}{1-\lambda}.
\end{displaymath}
Moreover,  $c$ can be set as $c = \frac{1}{1-\lambda}$.
\end{proof}

\section{Application of Kurtz's Theorem}
For reader's convenience,  we present a formal statement of the Kurtz's theorem as below.
\begin{theorem}
\textbf{[Kurtz]} Suppose we have a density dependent family satisfying the Lipschitz condition
\begin{displaymath}
\mid F(x)  - F(y)\mid \le M\|x-y \|
\end{displaymath} 
form some constant $M$. Further suppose $\lim_{n\to \infty} X(0) = x_0$, and let $x(t)$ be the deterministic process:
\begin{displaymath}
x(t) = x_0 +\int_{0}^t F(x(u) ) du,\  t\ge 0.
\end{displaymath}
Consider the path $\{X(u): u\le t\}$ for some fixed $t \ge 0$, and assume that there exists a neighborhood $K$ round this path satisfying 
\begin{displaymath}
\sum_{\textit{l} \in \textbf{L}} |\textit{l}| \sup_{x\in K}\beta_{\textit{l}}(x) <\infty.
\end{displaymath}
Then 
\begin{displaymath}
\lim_{N\to \infty} \sup_{u\le t} \mid X_N(u) -x(u)\mid = 0 \ a.s.
\end{displaymath}
\end{theorem}

In what follows, we show that $F(s)$ in our model is Lipschitz.
Recall that
\begin{equation}
\begin{split}
F(s) = \sum_{i=1}^\infty \left( \lambda p_{i-1}(s_{i-T-1}+ s_{i+T}) -p_i\right).
\end{split}
\end{equation}
Let $\textbf{x}$ and $\textbf{y}$ be two states, and let $p^x_{i-1} = x_{i-1}- x_i$ and $p^y_{i-1} = y_{i-1} - y_i$. Then, we have
\begin{equation}
\small
\nonumber
\begin{split}
&|F(\textbf{x}) -F(\textbf{y}) | =\mid  \sum_{i=1}^\infty \{ \lambda p^x_{i-1}(x_{i-1-T} + x_{i+T}) - (x_i-x_{i+1})\}\\
- & \sum_{i=1}^\infty \{ \lambda p^y_{i-1}(y_{i-1-T} + y_{i+T}) - (y_i-y_{i+1})  \mid \\
& \le \mid  \sum_{i=1}^\infty \{ \lambda p^x_{i-1}(x_{i-1-T} + x_{i+T} -y_{i-1-T} -y_{i+T})\\
+& \lambda (p^x_{i-1} -p^y_{i-1})(y_{i-1-T} + y_{i+T}) \}  \mid + \mid \sum_{i=1}^\infty(p^x_i -p^y_i)\mid   \\
&\le 2\lambda \sum_{i=0}^\infty |x_i -y_i|  + 2\lambda  \sum_i^{\infty} |x_{i-1} -x_{i } -y_{i-1} + y_{i}|  + 2 \mid \sum_{i=1}^\infty(x_i -y_i)\mid \\
& \le ( 6\lambda +2) \sum_{i=0}^{\infty} |x_i - y_i|  = ( 6\lambda +2)\|x- y\|_1.
\end{split}
\end{equation}


\end{document}